\documentclass[12pt,a4paper]{article}
\usepackage{amsfonts, enumerate, amsmath, amsthm, mathrsfs}

\usepackage{natbib}
\setcitestyle{authoryear,open={(},close={)}} %Citation-related commands

\usepackage{verbatim}

\newtheorem{theorem}{Theorem}

\newtheorem{lemma}{Lemma}
\newtheorem{proposition}{Proposition}
\newtheorem{definition}{Definition}

\def\d{{\rm d}}
\def\E{{\rm E}}

\def\P{{\rm P}}
\def\wt{\widetilde}

\newcommand{\var}{\textrm{Var}}
\newcommand{\ind}{\mathbf{1}}
\newcommand{\BR}{\mathbb{R}}

\newcommand{\mean}{\mathbb{E}}
\usepackage{graphicx, color}

\usepackage{mathtools}

\newcommand{\bx}{\boldsymbol{x}}
\newcommand{\bR}{\boldsymbol{R}}
\newcommand{\ba}{\boldsymbol{a}}
\newcommand{\bS}{\boldsymbol{S}}
\newcommand{\bX}{\boldsymbol{X}}

\begin{document}

\title{Posterior Inference via Hill's Prediction Model}
\author{
Pier Giovanni Bissiri\footnote{University of Milano-Bicocca, Italy, email: pier.bissiri@unimib.it}, 
Chris Holmes\footnote{University of Oxford, UK, email: chris.holmes@stats.ox.ac.uk} \,\,\&  Stephen G Walker\footnote{University of Texas at Austin, USA,  email: s.g.walker@math.utexas.edu}}

\date{}

\maketitle

\begin{abstract}This paper is concerned with the construction of prior free posterior distributions which rely on the use of one step ahead predictive distribution functions. These are typically more straightforward to motivate than prior distributions. Recent interest has been with Hill's $A_n$ prediction model through what has become known as conformal prediction. This model predicts the next observation to lie with equal probability in the intervals created by the observed data. The prediction model generates complete data sets which can be used to provide posterior inference on any statistic of  interest. 
\end{abstract}


\section{Introduction}

Statistics is primarily concerned with uncertainty and quantifying the uncertainty using probability models. It is therefore imperative to properly define the source of uncertainty so as to know what to model with probabilities. It should be clear that uncertainty is not being created by what has been observed. On the other hand, it is clear that uncertainty is created by what has not been seen. To elaborate on this point, what has not been seen is defined as the further information, usually represented by a sample size, which if seen the  answer to the statistical problem  becomes known to a satisfactory degree of accuracy. 

To quantify the uncertainty it becomes necessary to model the unseen data given the observed data. While this may be assisted by an assumed model for the data, it is not the main target. A predictive model using one step ahead  predictive density functions is all that is needed to provide a model for $X_{n+1:N}$ given $X_{1:n}$ for any $N>n$.

Both Bayesian and Frequentist approaches rely heavily on a likelihood function. 
Our perspective is fundamentally different. Namely, we do not model the observed data; rather, we build predictive models for the unobserved data needed for inference, which can be written as 
\begin{equation}\label{predictive}
p_{x_{1:n}}(x_{n+1:N}), 
\end{equation}
where \( p_{x_{1:n}} \) denotes a joint predictive model for the future observations and depends on the observed data. The joint predictive model, \( p_{x_{1:n}} \),  depends on the data but it is not a conditional probability, there is no $\mid$ as there is no joint model; i.e. 
\begin{equation}\label{noBayes}
p_{x_{1:n}}(x_{n+1:N}) \not\equiv  p(x_{1:N})/p(x_{1:n}),
\end{equation}
since we are not defining a $p(x_{1:n})$. Indeed, the inferential and statistical problem commences once the data has been observed. This makes our approach fundamentally different to
all likelihood based approaches, Bayesian and Frequentist, which rely on the preparation of a likelihood function, or a set of possible likelihood functions, as part of a protocol or simply as a device to avoid double use of the data in which a model is selected and then estimated from a single data set. Avoidance of this concern usually involves splitting the data into two parts to both estimate and validate models.

The lack of a conditional model (\ref{noBayes}) separates what we are doing very much from classical Bayes. For classical Bayes there is a heavy dependence on joint distributions for not only the observed data but future observations, ensuring the entire sequence is exchangeable. So, for example, the predictive
$p(x_{n+1}\mid x_{1:n})$ is a proper conditional distribution, in the sense that the $X_{n+1}$ depends on $X_{1:n}$.
On the other hand, our predictive is such that it is the density function itself, i.e. $p_{x_{1:n}}$ which depends on the observed data.  We can do this without imposing any structural properties on the $X_{1:n+1}$ since the focus should be on the predictive density function.

One of the original papers to use predictive models for Bayesian style inference was
\cite{Fong2023}, which is based on the original ideas appearing in \cite{Doob1949}. Suppose $X_{1:n}$ has been observed. Uncertainty from a statistical perspective arises from the unseen $X_{n+1:\infty}$; i.e. whatever is unknown would become known if the full data set has been seen. To quantify this uncertainty using probability a model $p_{x_{1:n}}(x_{n+1:N})$ is constructed for all lengths $N$. It is most convenient to do this using one-step ahead predictive density functions, so
$p_{x_{1:n}}(x_{n+1:N})=\prod_{m=n+1}^N p_{x_{1:n}}(x_m\mid x_{n+1:m-1}).$
The general idea is to sample $X_{n+1:N}$ and then present the statistic of interest, say $T=T(X_{1:N})$, which automatically has a posterior distribution given $X_{1:n}$. To implement this strategy one only needs to be able to sample from the chosen $p_{x_{1:n}}(x_m\mid x_{n+1:m-1})$, for any length $m>n$. 

A well known predictive model is the P\'olya-urn scheme, see \cite{Ferguson1973} and \cite{Blackwell1973}. For this model,  the one-step ahead predictive distribution is the empirical distribution function. So the sampling of $X_{n+1:N}$ would involve sampling $X_m$ from the empirical with observations $X_{1:m-1}$ and then the empirical distribution gets updated with the new observation $X_m$. This is precisely the idea of sampling a bag of balls of different colors, replacing the one chosen, and then adding another of the same color. This is known as the P\'olya-urn scheme and the sequence $X_{n+1:\infty}$ is exchangeable. 

However, for the existence of a posterior distribution of a statistic of interest, for example, the mean, it is only {\sl asymptotic exchangeability} that is required. That is, for some exchangeable sequence $(Z_i)_{i=1}^\infty$, 
$$(X_{m+1},X_{m+2},\ldots)\to (Z_1,Z_2,\ldots)$$
in distribution as $m\to\infty$. The posterior distribution is the de Finetti measure for the exchangeable sequence $(Z_i)$ which is connected to the sequence $(X_{m})_{m>n}$. Given that the sequence is constructed from a joint model which depends on the observed data, the asymptotic measure will depend on the observed data and hence can be described as a posterior measure. The goal therefore is to ensure the aforementioned sequence $(X_m)_{m>n}$ is asymptotically exchangeable. For more on asymptotic exchangeability, see \cite{Aldous85}.

The P\'olya-urn model does not produce  new samples, the sequence $X_{n+1:N}$ is comprised of the observed samples $X_{1:n}$. It can be shown that the limiting empirical distribution function is given by
$F(x)=\sum_{i=1}^n w_i\,\ind_{[X_i,\infty)}(x),$
where the random weights $(w_i)$ are from a Dirichlet distribution with parameters all set to 1 and \(\ind_A(x)=1\) if \(x\in A\) and 
\(\ind_A(x)=0\) if \(x\notin A\). Posterior inference for this model is known and was introduced in \cite{Rubin1981}, and is the Bayesian bootstrap. A generic use of this predictive model has the shortcoming in that no new sample values are created.

A more satisfying model would allow the samples to be continuous and yet allow the predictive distributions to be as fundamental in nature as the empirical distribution function is.  This is provided by conformal predictive distributions;  with the original idea being developed by \cite{Hill68}.  To see the fundamental idea, let $X_1,\ldots,X_n$ be an independent and identically distributed sequence with distribution $F$ and let $(X_{(i)})$ be the ordered values. Then it is well known that if $X$ is distributed as $F$, independently of the $(X_i)$, then 
$$E\left[P\bigg(X\in (X_{(i-1)},X_{(i)})\bigg)\right]=\frac{1}{n+1}$$
for $i=1,\ldots, n+1$, where $X_{(0)}=-\infty$ and $X_{(n+1)}=\infty$. Note the probability is acting on \(X\) and and the expectation on the  $(X_{(i)})$.
Hence, a predictive model for which $P(X_{n+1}\in (X_{(i-1)},X_{(i)}))=1/(n+1)$ for all $i$ is well motivated. The form of prediction described has become known as {\sl conformal prediction} and has now developed into a large area of research. See, for example, \cite{Gammerman98}, \cite{Vovk08}, \cite{Tocca2019}, \cite{Angel21} and \cite{Vovk22}. 

In section 2 we describe the framework for posterior distributions being constructed from predictive models. We do this using a specific example of a mean. In section 3 we detail the Hill prediction model and demonstrate the asymptotic exchangeability of the sequence $X_{n+1:\infty}$ conditional on the observed $X_{1:n}$. Section 4 provides the extension to the multivariate models which rely on copulas. Section 5 contains illustrations.

\section{The framework}

Suppose interest is in learning about $\theta$ which can be learnt completely from an infinite sample $X_{1:\infty}$ via the sequence $T_n=T(X_{1:n})$ and $T_n\to T_\infty=\theta^*$ almost surely. Here $\theta^*$ is the unknown true value. The assumption is that $T_n$ has been observed for some finite $n$. It goes without saying there is no uncertainty about the value of $T_n$ yet there remains uncertainty about the value of $T_\infty$. To quantify the uncertainty it is necessary to construct a probability model $p_{x_{1:n}}(T_\infty)$. This could well be constructed most conveniently by setting
$p_{x_{1:n}}(x_{n+1:N})$ for all $N>n$ and making use of one step ahead predictive density functions,
$p_{x_{1:n}}(x_{n+1:N})=\prod_{i=n+1}^N p_{x_{1:n}}(x_i\mid x_{n+1:i-1}).$
The distribution of $T_N$ given $x_{1:n}$ follows from
$T_N=T(X_{1:N})$.
In order for the existence of limits, one needs the existence of the limit of the distribution of $T_N$ as $N\to\infty$. This yields the posterior $p_{x_{1:n}}(T_\infty)$ which is the posterior for $\theta^*$. For a tidying exercise, asymptotic consistency based on an increasing sample size $n$ would ask that the $p_{T_n}(T_\infty)$ converges to a point mass at $\theta^*$ as $n\to\infty$. 

As an illustration of this, consider the normal model $N(x\mid\theta,1)$. Define $T_n=\bar{X}_n$, the sample mean, and
based on a one step ahead predictive 
$p_{x_{1:n}}(x_m\mid x_{n+1:m-1})=N(x_m\mid \bar{x}_{m-1},1)$
we get 
$$p(T_N\mid T_n)=N\left(T_N\mid T_n,1+\sum_{i=n+1}^{N-1}1/i^2\right).$$
The limit exists as $N\to\infty$ since the variance converges.
Then as $n\to\infty$ the variance converges to 0 and the mean converges to $\theta^*$.

When no such simple predictive exists, the reliance is on a sequence of one step ahead predictive density functions
$p_{x_{1:n}}(x_m\mid x_{n+1:m-1})$
for all $m>n$. In order for a posterior to exist, the most straightforward condition to impose is that the sequence $(X_{n+1},X_{n+2},\ldots)$ be asymptotically exchangeable.
This is equivalent to the sequence of random distributions
$P_{x_{1:n}}(x_m\mid x_{n+1:m-1})$ converging to a random distribution function. Hence, asymptotically the $(X_m)$ are independently and identically distributed from this limit random distribution and hence are asymptotically marginally exchangeable and which in turn implies the existence of a posterior measure, given $X_{1:n}$.

Returning to the normal illustration, from the one step ahead predictive it would need to be shown that the $\bar{X}_N$ converges almost surely to some random value. Indeed,
$\bar{X}_{m+1}=(m\bar{X}_m+X_{m+1})/(m+1)=X_m+z_m/(m+1),$
where $(z_m)$ is a sequence of independent and identically distributed standard normal random variables. Using Doob's martingale theorem it is clear that $(\bar{X}_m)$ converges almost surely. And hence for large $N$ the sequence $(X_{N+1},X_{N+2},\ldots)$ converges to an exchangeable sequence.

\section{Hill's prediction model}

Denote by \(x_1,\dotsc,x_n\) the observed data and assume that \(x_i \in [0,1]\) for \(i=1,\dotsc,n\).
Following \cite{Hill68}, the predictive value \(X_{m+1}\), $m\geq n$, is equally likely 
to fall in any of the open intervals between the successive order statistics 
of \(x_1,\dotsc,x_m\). From such intervals, it is natural to let \(X_{m+1}\) be uniformly distributed in the chosen interval.
In this model, the predictive cumulative distribution function of \(X_{m+1}\) is the ogive (or cumulative frequency polygon) based on \(x_1,\dotsc,x_m\).

To describe this mathematically, 
we assume that all random variables considered in this paper are defined on the some probability space \((\Omega, \mathscr{F},\P)\). Moreover, given the sequence \(X_1,X_2,\dotsc\), 
denote by \(X^m_{(j)}\) 
the \(j\)-th order statistic (namely the \(j\)--th smallest value) among \(X_1,\dotsc,X_m\), 
for \(j=1,\dotsc,m\) and let \(X^m_{(0)}=0\) 
and \(X^m_{(m+1)}=1\), for \(m=n,n+1,\dotsc\)
The predictive density function is:
\[
f_{X_{m+1}\mid X_{1:m}}(x) =  
\frac{1}{m+1}\,\sum_{j=1}^{m+1}
\ind\left[X^m_{(j-1)}, X^m_{(j)}\right)(x)\,
\frac{1}{X^m_{(j)}-X^m_{(j-1)}},
\]
where \(\ind_A(x)=1\) if \(x\in A\) and 
\(\ind_A(x)=0\) if \(x\notin A\), for every \(x\in[0,1]\) and every measurable subset \(A\) of \([0,1]\). 
The predictive cumulative distribution function is
\begin{equation}\label{eq: pred}
%\begin{split}
\P(X_{m+1}\leq x \mid X_{1:m})
=\frac{1}{m+1}\,\sum_{j=1}^{m+1} 
\ind\left[X^m_{(j-1)}, X^m_{(j)}\right)(x)\,
\left\{j-1+\frac{x-X^m_{(j-1)}}{X^m_{(j)}-X^m_{(j-1)}}\right\}
%\end{split}
\end{equation}
for every \(x\in[0,1)\). Denote by \(F_m\) the empirical cumulative distribution function based on \(X_{1:m}\), namely,
\begin{equation}\label{eq: ecdf}
F_m(x)=\frac{1}{m}\sum_{i=1}^m \ind(X_i\leq x),
\end{equation}
for every \(x\in [0,1)\).
Noting that \(X^m_{(j-1)} \leq x < X^m_{(j)}\) if and only if 
the number of observations smaller than or equal to \(x\) among \(X_1,\dotsc, X_m\) is exactly \(j-1\), \eqref{eq: pred} becomes:
\begin{equation}\label{eq: pred2}
%\begin{split}
P(X_{m+1}\leq x \mid X_{1:m})
=\frac{m\,f_m(x)}{m+1}
+ \frac{1}{m+1}\sum_{j=1}^{m+1} 
\ind_{[X^m_{(j-1)}, X^m_{(j)})}(x)\,
\left\{\frac{x-X^m_{(j-1)}}{X^m_{(j)}-X^m_{(j-1)}}\right\}
%\end{split}
\end{equation}
for every \(x\in[0,1)\)

\begin{theorem}\label{th: main}
The sequence of random variables \(X_{n+1:\infty}=(X_{1},X_{2},\dotsc)\)	
considered above is asymptotically exchangeable.	
\end{theorem}

The proof of the Theorem is based on the following Lemma whose proof is reported to make the present paper self-contained.
We write \(Y_m\stackrel{\textrm{d}}{\to}Y\), as \(m\to \infty\), if 
\(Y_{1:\infty}\) is a sequence of random variable that converges in distribution to \(Y\), being \(Y\) a random variable. 

\begin{lemma}\label{lemma}
If \(Y_{1:\infty}\) is a sequence of random variables valued into \([0,1]\) such that \(\lim_{m\to \infty} \mean(Y_m^r)=c_r \) for \(r=1,2,\dotsc\) and some sequence \(c_{1:\infty}\), 
then there exists a random variable \(Y\) such that 
\(Y_m\stackrel{\textrm{d}}{\to}Y\), as \(m\to \infty\), and \(\mean(Y^r)=c_r\), for \(r=1,2,\dotsc\).  
\end{lemma}

\begin{proof}[Proof of Lemma \ref{lemma}]
Let \(\nu_m\) be the probability distribution of \(Y_m\), for \(m=1,2,\dotsc\) Since each \(\nu_m\) is defined on \([0,1]\), which is compact, the sequence \(\nu_{1:\infty}\) is tight. If \(\nu_{m_k}\) is a subsequence that weakly converges to some probability measure \(\nu\) then by %the dominated convergence theorem, 
definition of weak convergence,
\[\lim_{k\to \infty} \int_0^1 x^r \d \nu_{m_k}(x)= 
\lim_{k\to \infty} \int_0^1 x^r \d \nu(x), \] 
but by hypothesis the same limit is equal to \(c_r\), for \(r=1,2,\dotsc\) 
Recalling that every probability distribution on \([0,1]\) is uniquely characterized by its moments, we have that every subsequence that weakly converges to some probability measure needs to converge to \(\nu\) and therefore \citep[see corollary to Theorem 25.10,][]{Billingsley95} \(\nu_m\) weakly converges to \(\nu\) and the proof is complete.

\end{proof}

\begin{proof}[Proof of Theorem \ref{th: main}]
Let \(\mu_m=\sum_{j=1}^m \delta_{X_j}/m\) be the empirical measure for \(m=1,2,\dotsc\) 
To begin with, let us prove that there exists a random probability measure \(P_\infty\) such that \(\mu_m \stackrel{\textrm{w}}{\to}P_\infty\), as \(m\to+\infty\), almost surely, where \(\stackrel{\textrm{w}}{\to}\) denotes weak convergence for probability measures.  
By \eqref{eq: pred2}, we have that
\begin{equation}\label{eq: ineq}
\frac{m}{m+1}F_m(x) \leq P(X_{m+1}\leq x \mid X_{1:m})
\leq \frac{m}{m+1}F_m(x) +\frac{1}{m+1},
\end{equation}
which in turn implies that:
\begin{equation}\label{eq: ineq2}
%\begin{split}
\frac{m}{m+1}(1-F_m(x)) \leq P(X_{m+1}> x \mid X_{1:m})
\leq \frac{m}{m+1}(1-F_m(x)) +\frac{1}{m+1},
%\end{split}
\end{equation}
Recalling that for any probability measure \(\nu\) on \([0,1]\) 
with cdf \(G_\nu\) we have that:
\[
\int_{0}^1 x^r \d \nu(x)= \int_0^1 (1-G_\nu(x^{1/r})) \d x,
\]
so inequality \eqref{eq: ineq2} yields:
\begin{equation}\label{eq: ineq3}
\frac{m}{m+1}\,M_{m,r} \leq \mean(X_{m+1}^r \mid X_{1:m})
\leq \frac{m}{m+1}\,M_{m,r} +\frac{1}{m+1},
\end{equation}
where:
\[
M_{m,r} = \int_{0}^1 x^r \d \mu_m(x)
= \frac{1}{m} \sum_{j=1}^m X_j^r,
\]
for \(r=1,2,\dotsc\)
At this stage, note that, for \(r=1,2,\dotsc\)
\begin{equation*}
M_{m+1,r}= \frac{m}{m+1}\,M_{m,r}+\frac{1}{m+1}\,X_{m+1}^r,
\end{equation*}
and therefore
\begin{equation*}
\mean(M_{m+1,r}\mid X_{1:m})
= \frac{m}{m+1}\,M_{m,r}+\frac{1}{m+1}\,\mean(X_{m+1}^r\mid X_{1:m}),
\end{equation*}
which by the second inequality in \eqref{eq: ineq3} yields:
%\begin{equation*}
$\mean(M_{m+1,r}\mid X_{1:m})
\leq M_{m,r}+1/(m+1)^2.$
%\end{equation*}
By the Robbins and Siegmund Theorem 
\citep{Robbins71}, the random sequence 
\(M_{m,r}\) converges to a random variable almost surely, 
as \(m\to \infty\), for \(r=1,2,\dotsc\) 
By Lemma \ref{lemma}, this implies that 
there exists a random probability measure \(P_\infty\) such that \(\mu_m \stackrel{\textrm{w}}{\to}P_\infty\). 
If we denote by \(F_\infty\) the (random) cdf of \(P_\infty\), 
we have that \(F_m(x)\) converges to \(F_\infty(x)\) for every continuity point \(x\) of \(F_\infty\), almost surely. 
By \eqref{eq: ineq}, we have that the same is true for 
the predictive cdf \(P(X_{m+1}\leq x \mid X_{1:m})\) and therefore 
\(P_m\stackrel{\textrm{w}}{\to}P_\infty\), almost surely, being \(P_m\) the predictive distribution of \(X_{1:\infty}\), namely:
$
P_m(B)=\P(X_{m+1}\in B \mid X_{1:m}),
$
for every $m$ and Borel subset \(B\) of \([0,1]\). 
	
Finally, recall that a sequence of random variables is asymptotically exchangeable if 
the corresponding sequence of predictive distributions weakly converges to a probability measure, almost surely \citep[see Lemma 8.2][]{Aldous85} and the proof is complete. 
\end{proof}

\section{Bivariate and multivariate extensions}

We start with the bivariate case. 
Let \(x_1,y_1, \dotsc, x_n, y_n \) be the data and assume that \(x_\ell,y_\ell \in [0,1]\) for \(\ell=1,\dotsc,n\). As before, for the theory, we can let \(X_1=x_1,Y_1=y_1, \dotsc, X_n=x_n,Y_n=y_n\) and we can use Hill's predictive model to assess the conditional distribution of \(X_{m+1}\) given \(X_{1:m}\) and of \(Y_{m+1}\) given \(Y_{1:m}\), for \(m=n,n+1,\dotsc\), but how to assess the (joint) conditional distribution of \((X_{m+1},Y_{m+1})\) given \(X_{1:m},Y_{1:m}\)? We use copulas.
 
In virtue of Hill's model, we are ensuring that the conditonal distribution functions \(F_m\) and \(G_m\) of \(X_{m+1}\) given \(X_{1:m}\) and of \(Y_{m+1}\) given \(Y_{1:m}\) are, respectively, continuous. We let \(X_{m+1}\) and \(Y_{1:m}\) be conditionally independent given \(X_{1:m}\) and \(Y_{m+1}\) and \(X_{1:m}\) be conditionally independent given \(Y_{1:m}\) so that \(F_m\) and \(G_m\) are the conditional distribution functions of 
\(X_{m+1}\) and of \(Y_{m+1}\) given \((X_{1:m},Y_{1:m})\), respectively. 
In this way, if we let \(U_{m+1}=F_{m}(X_{m+1})\) and 
\(V_{m+1}=G_m(Y_{m+1})\) then \(U_{m+1}\) and \(V_{m+1}\) are uniform random variables on \([0,1]\) conditionally upon \(X_{1:m},Y_{1:m}\). 

We can consider 
the Gaussian copula \(C_\rho\)  
%\begin{equation*}
$C_\rho(u,v)=\Phi_\rho(\Phi^{-1}(u),\Phi^{-1}(v))$	
%\end{equation*}
where \(\Phi\) is the standard univariate Gaussian distribution function and 
\(\Phi_\rho\) is the joint distribution function of a bivariate Gaussian random vector with standardized marginals and correlation coefficient \(\rho\). 
Let the distribution of \((U_{m+1},V_{m+1})\) be 
\(C_{R_m}\), where \(R_m\) is a function of \((X_{1:m},Y_{1:m})\). 
So, if we let \(\tilde{X}_{m+1}=\Phi^{-1}(U_{m+1})\) and \(\tilde{Y}_{m+1}=\Phi^{-1}(V_{m+1})\), where \(\Phi\) is the standardized Gaussian distribution,  then \((\tilde{X}_{m+1},\tilde{Y}_{m+1})\) is a Gaussian random bivariate vector with standardized marginals and correlation coefficient \(R_m\).

We are now going to see how it is possible to assess the random sequence 
\(R_{n},R_{n+1},\dotsc\) and consequently also the random sequence 
\begin{equation}\label{eq: seq}
(\tilde{X}_{n+1},\tilde{Y}_{n+1}), (\tilde{X}_{n+2},\tilde{Y}_{n+2}), \dotsc 
\end{equation}
As it will become evident, it is very useful to ensure that the sequence \(R_{n},R_{n+1},\dotsc\) converges almost surely. 
As we shall see, this can be done at least in two ways. 
For \(\ell=1,\dotsc,n\), let \(r_{1,\ell}\) and \(r_{2,\ell}\) be the rank of \(x_\ell\) among \(x_1,\dotsc,x_n\) and of \(y_\ell\) among  \(y_1,\dotsc,y_n\), respectively, namely:
\[r_{1,\ell}=\sum_{j=1}^n \ind_{[0,x_\ell]}(x_j),
\qquad
r_{2,\ell}=\sum_{j=1}^n \ind_{[0,y_\ell]}(y_j)
\]
so that \(F_n(x_\ell)=r_{1,\ell}/(n+1)\) and \(G_n(y_\ell)=r_{2,\ell}/(n+1)\). Moreover, let \(\tilde{x}_\ell=\Phi^{-1}(r_{1,\ell}/(n+1))\),   \(\tilde{y}_\ell=\Phi^{-1}(r_{2,\ell}/(n+1))\).

\begin{definition}
We shall say that the distribution of the random sequence 
\eqref{eq: seq} falls under \emph{model A} if 
for \(m=n,n+1,\dotsc\), conditionally upon 
\((\tilde{X}_{1:m},\tilde{Y}_{1:m})\) the random pair 
\((\tilde{X}_{m+1},\tilde{Y}_{m+1})\)
is a Gaussian random vector with standardized marginals and correlation coefficient \(R_m\), and moreover,
\begin{align}
S^2_{n,1}&=\frac{1}{n}\sum_{\ell=1}^{n}\tilde{x}_\ell^2 , \nonumber	\\
S^2_{m+1,1}&= 
\frac{m}{m+1}\,S^2_{m,1} + \frac{1}{m+1}\,S^2_{m,1}\,\tilde{X}_{m+1}^2, \label{eq: Sn1}\\
S^2_{n,2}&=\frac{1}{n}\sum_{\ell=1}^{n}\,\tilde{y}_\ell^2 ,	\nonumber \\
S^2_{m+1,2}&= 
\frac{m}{m+1}\,S^2_{m,2} + \frac{1}{m+1}\,S^2_{m,2}\,\tilde{Y}_{m+1}^2, \label{eq: Sn2}\\ 
S_n&=\frac{1}{n}\sum_{\ell=1}^{n}\tilde{x}_\ell\,\tilde{y}_\ell , \nonumber \\
S_{m+1}&= 
\frac{m}{m+1}\,S_m + \frac{1}{m+1}\,\sqrt{S^2_{m,1}\,S^2_{m,2}}\,\tilde{X}_{m+1}\,\tilde{Y}_{m+1}, 
\label{eq: Sm} \\
\label{eq: Rm}
R_m &= \frac{S_{m}}{\sqrt{S^2_{m,1}\,S^2_{m,2}}},
\end{align}
for \(m=n,n+1,\dotsc\)
\end{definition}

\vspace{0.1in}
\noindent
Note that under model A,  
\begin{equation*}
\begin{split}
S_m&=\frac{1}{m}\sum_{\ell=1}^{n}\tilde{x}_\ell\,\tilde{y}_\ell+
\frac{1}{m}\sum_{\ell=n+1}^{m}\sqrt{S^2_{\ell-1,1}\,S^2_{\ell-1,2}}
\tilde{X}_\ell\,\tilde{Y}_\ell\\
S^2_{m,1}&= 
\frac{1}{m}\sum_{\ell=1}^{n}\tilde{x}_\ell^2+
\frac{1}{m}\sum_{\ell=n+1}^{m}S^2_{\ell-1,1}
\tilde{X}^2_\ell\\
S^2_{m,2}&=
\frac{1}{m}\sum_{\ell=1}^{n}\tilde{y}_\ell^2+
\frac{1}{m}\sum_{\ell=n+1}^{m}S^2_{\ell-1,2}
\tilde{Y}^2_\ell
\end{split}
\end{equation*}
and therefore by the Cauchy--Schwartz inequality we have that 
\(S_m^2\leq S^2_{m,1}S^2_{m,2}\) so that \(R_m\in[-1,1]\) for \(m=n,n+1,\dotsc\).

\begin{definition}
We shall say that the distribution of the random sequence 
\eqref{eq: seq} falls under \emph{model B} if 
for \(m=n,n+1,\dotsc\), conditionally upon 
\((\tilde{X}_{1:m},\tilde{Y}_{1:m})\) the random pair 
\((\tilde{X}_{m+1},\tilde{Y}_{m+1})\)
is a Gaussian random vector with standardized marginals and correlation coefficient \(R_m\), and moreover,
\begin{align}
S^2_{n,1}&=\frac{1}{n}\sum_{\ell=1}^{n}\tilde{x}_\ell^2 , \nonumber  \\
S^2_{m+1,1}&= 
\frac{m}{m+1}\,S^2_{m,1} + \frac{1}{m+1}\,\tilde{X}_{m+1}^2, \nonumber \\
S^2_{n,2}&=\frac{1}{n}\sum_{\ell=1}^{n}\,\tilde{y}_\ell^2 ,	\nonumber \\
S^2_{m+1,2}&= 
\frac{m}{m+1}\,S^2_{m,2} + \frac{1}{m+1}\,\tilde{Y}_{m+1}^2, \nonumber \\ 
S_n&=\frac{1}{n}\sum_{\ell=1}^{n}\tilde{x}_\ell\,\tilde{y}_\ell , \nonumber \\
S_{m+1}&= 
\frac{m}{m+1}\,S_m + 
\frac{1}{m+1}\,\tilde{X}_{m+1}\,\tilde{Y}_{m+1}, 
\label{eq: modelB-cov} \\
R_m &= \frac{S_{m}}{\sqrt{S^2_{m,1}\,S^2_{m,2}}}, \nonumber
\end{align}
for \(m=n,n+1,\dotsc\)
\end{definition}

\vspace{0.1in}
\noindent
Under model B, if we let \(\tilde{X}_\ell=\tilde{x}_\ell\), 
\(\tilde{Y}_\ell=\tilde{y}_\ell\), for \(\ell=1,\dotsc,n\), then for \(m=n,n+1,\dotsc\)
\begin{align}\label{eq: S-modelB}
S_m&=\frac{1}{m}\sum_{\ell=1}^{m}\tilde{X}_\ell\,\tilde{Y}_\ell, & 
S^2_{m,1}&= 
\frac{1}{m}\sum_{\ell=1}^{m}\tilde{X}_\ell^2, &
S^2_{m,2}&=
\frac{1}{m}\sum_{\ell=1}^{m}\tilde{Y}_\ell^2.
\end{align}
Therefore, applying again the Cauchy-Schwartz inequality, one can see that 
\(R_m\in[-1,1]\), for \(m=n,n+1,\dotsc\)

To sum up, for \(m=n,n+1,\dotsc\) we generate two standardized Gaussian random variables \(\tilde{X}_{m+1}\) and \(\tilde{Y}_{m+1}\) with correlation coefficient \(R_m\), taking for instance 
\(\tilde{X}_{m+1}=Z_{m+1,1}\) and 
\(\tilde{Y}_{m+1}=R_m Z_{m+1,1} + \sqrt{1-R_m^2}Z_{m+1,2}\), where 
\(Z_{m+1,1}\) and \(Z_{m+1,2}\) are iid standardized Gaussian random variables, 
 then let \(U_{m+1}=\Phi(\tilde{X}_{m+1})\), 
\(V_{m+1}=\Phi(\tilde{Y}_{m+1})\), and finally let \(X_{m+1}=F_m^{-1}(U_{m+1})\) and \(Y_{m+1}=G_m^{-1}(V_{m+1})\), where
\[
F^{-1}_m(u)=\sum_{\ell=1}^{n+1}
\ind_{\left(\frac{\ell-1}{m+1}, \frac{\ell}{m+1}\right]}(u)
\, \{X^{m}_{(\ell-1)}+(X^{m}_{(\ell)}-X^m_{(\ell-1)})
((m+1)u-\ell+1)\},
\]
\(X^{m}_{(1)},\dotsc,X^{m}_{(m)}\) are the order statistics of \(X_{1:m}\), \(X^{m}_{(0)}=0\), 
\(X^{m}_{(m+1)}=1\), 
and a similar expression holds for \(G^{-1}_m\).

Here are some comments about the two models A and B.  
For both models, the aim is to define a sequence of Gaussian random pairs \eqref{eq: seq} with standardized marginals,  \(R_m\) is the correlation coefficient of the \((m+1)\)-th pair \((\tilde{X}_{m+1},\tilde{Y}_{m+1})\) given the past and is equal to the estimated correlation coefficient on the basis of the past observations. 

Under model A, we generate a sequence of Gaussian random pairs \linebreak \((X'_{n+1},Y'_{n+1}),(X'_{n+2},Y'_{n+2}),\dotsc \) with zero mean, but variances  
not necessarily equal to one. %For each coordinate, this variance is the second moment of the past observations. 
The conditional correlation coefficient of \((X'_{m+1},Y'_{m+1})\) given the past is \(R_m\) and 
the conditional variances of \(X'_{m+1}\) and of \(Y'_{m+1}\) given the past are 
\[
S_{m,1}^2=\frac{1}{m}\sum_{\ell=1}^m (X_\ell')^2, \quad 
S_{m,2}^2=\frac{1}{m}\sum_{\ell=1}^m (Y_\ell')^2,
\quad 
\]
respectively, where \(X_\ell'=\tilde{x}_\ell\) and 
\(Y_\ell'=\tilde{y}_\ell\), for \(\ell=1,\dotsc,n\). 
Each element of this sequence of Gaussian random pairs is divided by the corresponding standard deviation to obtain a sequence of Gaussian random pairs with standardized marginals, namely:
\[
\widetilde{X}_{m+1}= X_{m+1}'/\sqrt{S_{m,1}^2}, 
\quad 
\widetilde{Y}_{m+1}= Y_{m+1}'/\sqrt{S_{m,2}^2}, 
\]
for \(m=n,n+1,\dotsc\)
%and \(R_m\) is the corresponding correlation coefficient. 

Under model B, the sequence of Gaussian random pairs \eqref{eq: seq} with standardized marginals is generated directly.
Simulations show that there is no substantial difference between these two models.

As anticipated, it is fundamental to ensure that the sequence  \(R_m\) converges, almost surely. Indeed, we can prove the following result:
\begin{proposition}\label{prop: convRm}
For model A and B, there exist four random variables \linebreak \(S_{\infty},S^2_{\infty,1},S^2_{\infty,2},R_\infty\)
such that
\begin{align*}
\lim_{m\to \infty} S^2_{m,1}&=S^2_{\infty,1},& 
\lim_{m\to \infty} S^2_{m,1}&=S^2_{\infty,1},\\
\lim_{m\to \infty} S_{m}&=S_{\infty}, &
\lim_{m\to \infty} R_{m}&=R_\infty=S_\infty/\sqrt{S^2_{\infty,1}S^2_{\infty,2}},
\end{align*}
almost surely. 
\end{proposition}

\begin{proof}

Let us start with model A. 
Let \(\mathscr{G}_m\) be the sigma-algebra generated by \(X_{1:m},Y_{1:m}\), for \(m=n+1,n+2,\dotsc\).
The random sequences \(S^2_{m,1}\), \(S^2_{m,2}\) and \(S_m\) are all martingales with respect to the filtration \(\mathscr{G}_m\). 
Since \(S^2_{m,1}\) and \(S^2_{m,2}\) are also nonnegative they converge almost surely. Almost sure convergence of \(R_m\) trivially follows from convergence of 
\(S^2_{m,1}\), \(S^2_{m,2}\) and of \(S_m\). So, all we need to show is convergence 
of \(S_m\). 
%The random sequence \(S_{n},S_{n+1},\dotsc\) is a martingale since 
%\begin{equation*}
%\E(S_{m+1}\mid \mathscr{G}_m)=S_m
%\end{equation*}
%by \eqref{eq: Sm}. 
To this aim, it is sufficient to prove that \(\var(S_m)\leq c\) for some constant \(c\) and \(m=n+1,n+2,\dotsc\). 
By \eqref{eq: Sm}, 
\begin{equation*}
\begin{split}
\var(S_{m+1}\mid  \mathscr{G}_m)&= 
\frac{S^2_{m,1}\,S^2_{m,2}}{(m+1)^2}\,\var(\tilde{X}_{m+1}\,\tilde{Y}_{m+1}\mid \mathscr{G}_m)\\
&= \frac{S^2_{m,1}\,S^2_{m,2}}{(m+1)^2}\, \var(R_m Z_{m+1,1}^2 + \sqrt{1-R_m^2}Z_{m+1,1}Z_{m+1,2}\mid \mathscr{G}_m)\\
&= \frac{S^2_{m,1}\,S^2_{m,2}}{(m+1)^2}\, (1+R^2_m)\\
&\leq \frac{2\,S^2_{m,1}\,S^2_{m,2}}{(m+1)^2}
%&= \frac{S^2_{m,1}\,S^2_{m,2}+S_m^2}{(m+1)^2}
\end{split}
\end{equation*}
and therefore:
\begin{equation}\label{eq: var_Sm}
\begin{split}
\var(S_{m+1}) &= 
\var(\E(S_{m+1}\mid  \mathscr{G}_m)) + \E(\var(S_{m+1}\mid  \mathscr{G}_m))\\
&\leq \var(S_m)+  2\,\E(S^2_{m,1}\,S^2_{m,2})/(m+1)^2.
\end{split}
\end{equation}
%since \(E(S_m)=S_n\) for every \(m=n+1,n+2,\dotsc\)
At this stage we need to compute \(\E(S^2_{m,1}\,S^2_{m,2})\) or find an upper bound for it. 
By the recursion formulas given in \eqref{eq: Sn1} and \eqref{eq: Sn2} we have: 
\begin{equation}\label{eq: stepvar1}
\begin{split}
S^2_{m+1,1}\,S^2_{m+1,2} 
&=\frac{S^2_{m,1}\,S^2_{m,2}}{(m+1)^2}\,
(m+\tilde{X}_{m+1}^2)(m+\tilde{Y}_{m+1}^2)\\
&=\frac{S^2_{m,1}\,S^2_{m,2}}{(m+1)^2}\,
\left\{m^2+m(\tilde{X}_{m+1}^2+\tilde{Y}_{m+1}^2)+
\tilde{X}_{m+1}^2\tilde{Y}_{m+1}^2\right\}.
\end{split}
\end{equation}
At this stage, note that 
\begin{equation*}
\E(\tilde{X}_{m+1}^2\,\tilde{Y}_{m+1}^2\mid \mathcal{G}_m)= 
\var(\tilde{X}_{m+1}\,\tilde{Y}_{m+1}\mid \mathcal{G}_m)=
1+R^2_m\leq 2,
\end{equation*}
and therefore \eqref{eq: stepvar1} yields:
\begin{equation*}
\E(S^2_{m+1,1}\,S^2_{m+1,2}\mid \mathcal{G}_m) 
\leq \left\{1+1/(m+1)^2\right\}\,S^2_{m,1}\,S^2_{m,2}
\end{equation*}
and therefore
%\begin{equation*}
$\E(S^2_{m+1,1}\,S^2_{m+1,2}) 
\leq \left\{1+1/(m+1)^2\right\}\,\E(S^2_{m,1}\,S^2_{m,2})$
%\end{equation*}
so that
\begin{equation}\label{eq: stepvar2}
\E(S^2_{m,1}\,S^2_{m,2}) \leq S^2_{n,1}\,S^2_{n,2}\prod_{\ell=n+1}^{m}
\left(1+\frac{1}{\ell^2}\right)
\leq S^2_{n,1}\,S^2_{n,2}\prod_{\ell=n+1}^{\infty}
\left(1+\frac{1}{\ell^2}\right)
\end{equation}
Note that
\begin{equation*}
%\begin{split}
\prod_{\ell=n+1}^{\infty}\left(1+\frac{1}{\ell^2}\right) 
\leq \exp\left\{\sum_{\ell=n+1}^{\infty} 
\log\left(1+\frac{1}{\ell^2}\right)\right\} 
\leq \exp\left\{\sum_{\ell=n+1}^{\infty} \frac{1}{\ell^2}\right\}
<\infty. 
%\end{split}
\end{equation*}
Let 
$
c= S^2_{n,1}\,S^2_{n,2}\prod_{\ell=n+1}^{\infty}
\left(1+\frac{1}{\ell^2}\right).
$
A combination of \eqref{eq: var_Sm} and \eqref{eq: stepvar2} yields
%\begin{equation*}
$\var(S_{m+1}) \leq \var(S_{m}) +2c/(m+1)^2$
%\end{equation*}
and therefore
\begin{equation*}
\begin{split}
\var(S_m)\leq  \sum_{\ell=n+1}^{m}\frac{2c}{\ell^2} \leq 
\sum_{\ell=n+1}^{\infty}\frac{2c}{\ell^2},
\end{split}	
\end{equation*}
for \(m=n+1,n+2, \dotsc\), 
and the proof is complete for model A.

Let us now consider model B. 
Recalling \eqref{eq: S-modelB}, the random sequence \(S^2_{m,1}\) converges almost surely to one (as \(m\to \infty\)) by the strong law of large numbers since \(\tilde{X}_{n+1},\tilde{X}_{n+2}, \dotsc\) are iid random variables. The same result holds for 
\(S^2_{m,2}\). 
So, \(S^2_{\infty,1}=S^2_{\infty,2}=1\), almost surely. 
As a consequence, we just need to prove almost sure convergence of \(S_m\). To this aim we first prove almost sure convergence of the square \(S_m^2\). 
By \eqref{eq: modelB-cov}, we have
\begin{equation*}
S_{m+1}^2= 
\left(\frac{m}{m+1}\right)^2\,S_m^2 +
\frac{2m}{m+1}\,S_m\,\tilde{X}_{m+1}\,\tilde{Y}_{m+1}
+ \left(\frac{1}{m+1}\right)^2\,
\tilde{X}_{m+1}^2\,\tilde{Y}_{m+1}^2.   
\end{equation*}
Being \(\E(\tilde{X}_{m+1}\tilde{Y}_{m+1}\mid 
\mathscr{G}_m)=R_m\) and 
\(\E(\tilde{X}^2_{m+1}\tilde{Y}^2_{m+1}\mid \mathscr{G}_m)=1+2R_m^2\), we get
\begin{equation}\label{eq: proof-modelB}
\begin{split}
\E(&S_{m+1}^2\mid \mathscr{G}_m)\\
&= \left(\frac{m}{m+1}\right)^2S_m^2+
\frac{2m}{(m+1)^2}S_mR_m
+\frac{1}{(m+1)^2}(1+2R_m^2)\\
&= \left(\frac{m}{m+1}\right)^2S_m^2+
\frac{2m}{(m+1)^2}\frac{S_m^2}{\sqrt{S_{m,1}^2S_{m,2}^2}}
+\frac{1}{(m+1)^2}(1+2R_m^2)\\
& \leq 
(1+B_m)\, S_m^2 +3/(m+1)^2,
\end{split} 
\end{equation}
where
\begin{equation*}
B_m =  \left(\frac{m}{m+1}\right)^2 
\left(1+\frac{2}{m\sqrt{S_{m,1}^2S_{m,2}^2}}\right)-1 .  
\end{equation*}
Since both \(S_{m,1}^2\) and \(S_{m,2}^2\) converge to one, 
almost surely, 
\[
B_m \sim  \left(\frac{m}{m+1}\right)^2
\left(1 +  \frac{2}{m} \right)-1
=  \frac{m(m+2)}{(m+1)^2}-1=-\frac{1}{(m+1)^2}<0, 
\] 
as \(m\to \infty\), almost surely. 
This implies that \(B_m <0\) for large \(m\), almost surely, 
which in turn implies that 
$\sum_m (B_m\vee 0) <\infty, $
almost surely, and by \eqref{eq: proof-modelB}, 
\[
\E(S_{m+1}^2\mid \mathscr{G}_m) 
\leq  \{1+(B_m\vee 0)\}S_m^2 + 3/(m+1)^2.
\] 
We can now apply Robbins-Siegmund Theorem \citep{Robbins71} to obtain that \(S^2_m\) converges almost surely. 
At this stage, note that \eqref{eq: modelB-cov} yields
%\begin{equation*}
$ S_{m+1} - S_{m} = 
(-\,S_m +\tilde{X}_{m+1}\,\tilde{Y}_{m+1})/(1+m)$,  
%\end{equation*}
and therefore:
\begin{equation}\label{eq: almost-final}
\lvert S_{m+1} - S_{m} \rvert \leq  
\frac{1}{m+1}\,\lvert S_m \rvert + 
\frac{1}{m+1}\, \lvert \tilde{X}_{m+1}\,\tilde{Y}_{m+1} \rvert. \end{equation}
Since \(\E(\tilde{X}^2_{m+1}\tilde{Y}^2_{m+1}\mid \mathscr{G}_m)=1+2R_m^2\), by Markov's inequality, for every \(\varepsilon>0\), we have
\begin{equation*}
\P\left(
\frac{\lvert \tilde{X}_{m+1}\,\tilde{Y}_{m+1} \rvert}{m+1}  >\varepsilon\right) \leq \frac{\E(\E(\tilde{X}^2_{m+1}\tilde{Y}^2_{m+1}\mid \mathscr{G}_m))}{(m+1)^2}\leq \frac{3}{(m+1)^2}
\end{equation*}
so that by the First Borel--Cantelli Lemma, 
\(\lvert \tilde{X}_{m+1}\,\tilde{Y}_{m+1} \rvert/(m+1)\) converges to zero, almost surely. Moreover, since \(S_m^2\) converges, almost surely, so does \(\lvert S_m \rvert\) and therefore by \eqref{eq: almost-final}, 
\begin{equation} \label{eq: lim-proof}
\lim_{m\to \infty} \lvert S_{m+1} - S_{m} \rvert =0, 
\end{equation}
almost surely. 
The sequence \(\lvert S_m \rvert\) converges to a nonnegative random variable, say \(W\). If \(W=0\), then \(S_m\) almost surely converges to zero, as well. If \(W>0\), then for every \(\varepsilon>0\),  
\begin{equation*}%\label{eq: almost-there}
S_m \in [-W-\varepsilon, -W+\varepsilon]\cup 
[W-\varepsilon, W+\varepsilon],
\end{equation*}
for large \(m\), almost surely. But by \eqref{eq: lim-proof}, 
$\lvert S_{m+1} - S_{m} \rvert \leq W-\varepsilon,
$
for large \(m\), almost surely. So, if \(S_m\leq -W+\varepsilon\) then \(S_{m+1} \leq 0\), and if 
\(S_m\geq W-\varepsilon\) then \(S_{m+1} \geq 0\), for large \(m\), almost surely. 
This implies that the sign of \(S_m\) converges and given that 
\(\lvert S_m \rvert\) converges, we have that 
the sequence \(S_m\) converges, almost surely, and the proof is complete. 
\end{proof}

\vspace{0.1in}
\noindent
We are now ready to prove the following:
\begin{theorem}\label{th: bivariate-hill}
The random sequence \((X_{1},Y_{1}), (X_{2},Y_{2}),\dotsc\) 
considered above is asymptotically exchangeable under model A and under model B. 
\end{theorem}

\begin{proof}[Proof of Theorem \ref{th: bivariate-hill}]
	
In our framework, the conditional cdf of \((X_{m+1},Y_{m+1})\) given 
\(X_{1:m},Y_{1:m}\) is:
\begin{equation}\label{eq: jointpred} 
\P(X_{m+1}\leq x,Y_{m+1}\leq y \mid  X_{1:m},Y_{1:m}) = C_{R_m}(F_m(x), G_m(y))	
\end{equation}
By Proposition \ref{prop: convRm}, \(R_m\) converges almost surely to \(R_\infty\) as \(m\to \infty\). 
Therefore, the sequence of copulas \(C_{R_m}\) pointwise converges to \(C_{R_\infty}\), almost surely. 
By Theorem \ref{th: main}, we have that \(F_m\) and \(G_m\) weakly converge to 
a random cdf, almost surely. Let \(F_\infty\) and \(G_\infty\) be the limit of 
\(F_m\) and \(G_m\), respectively. 
By Theorem 2 in \cite{Sempi04} we have that the cdf in \eqref{eq: jointpred} weakly converges to \(C_{R_\infty}(F_\infty(x),G_\infty(y))\). 
The thesis follows from Lemma 8.2 of \cite{Aldous85}. 
\end{proof}

\vspace{0.1in}
\noindent
It is not difficult to see how the proposed scheme can be extended from the bivariate case to the multivariate case. 
We shall show how to extend model A. Mutatis mutandis, model B can be extended to the multivariate case, as well. 
Let \(\bx_1,\dotsc,\bx_n\) be our data, where \(\bx_\ell=(x_{\ell,1},\dotsc,x_{\ell,d})\in [0,1]^d\), for \(\ell=1,\dotsc,n\) 
and some positive integer \(d\). As before, let 
\(X_{\ell,k}=x_{\ell,k}\), for \(\ell=1,\dotsc,n\) and \(k=1,\dotsc,d\). 
Let \(F_{m,k}\) be the marginal cdf of \(X_{m+1,k}\) given \(X_{1,k},\dotsc,X_{m,k}\), for \(k=1,\dotsc,d\) obtained using Hill's predictive model. For \(k=1,\dotsc,d\), let \(\tilde{x}_{\ell,k}=\Phi^{-1}(r_{k,\ell}/(n+1))=\Phi^{-1}(F_{n,k}(x_{\ell,k}))\), 
where \(r_{k,\ell}\) is the rank of \(x_{\ell,k}\) among 
\(x_{1,k},\dotsc,x_{n,k}\). Moreover, for \(k,h=1,\dotsc,d\), we can extend model A to the multivariate case as follows:
\begin{equation*}
\begin{split}
S_{n,k,h}&=\frac{1}{n}\sum_{\ell=1}^{n}\tilde{x}_{\ell,k}\tilde{x}_{\ell,h} ,	\\
S_{m+1,k,h}&= 
\frac{m}{m+1}\,S_{m,k,h} + \frac{1}{m+1}\,
\sqrt{S_{m,k,k}S_{m,h,h}}\tilde{X}_{m+1,k}\tilde{X}_{m+1,h}, \\
R_{m,k,h}&=\frac{S_{m,k,h}}{\sqrt{S_{m,k,k}\,S_{m,h,h}}}
\end{split}
\end{equation*}
where \((\tilde{X}_{m+1,1}, \dotsc,\tilde{X}_{m+1,d})\) is a multivariate Gaussian random vector with standardized marginals and correlation matrix 
$\bR_m=(R_{m,k,h})_{k,h=1,\dotsc,d},$
and 
\(X_{m+1,k}= F^{-1}_{m,k}(\Phi(\tilde{X}_{m+1,k}))\),  
for \(k=1,\dotsc,d\) and \(m=n,n+1,\dotsc \).
For model A, the random variables \(\tilde{X}_{m+1,k}\), for \(k=1,\dotsc,d\) and \(m=n,n+1,\dotsc\), can be generated as follows:
\begin{equation*}
\tilde{X}_{m+1,k} = 
S_{m,k,k}^{-1/2}\left\{ 
\frac{1}{m} \sum_{\ell=1}^n \tilde{x}_{\ell,k}\,Z_{m,\ell}
+ \frac{1}{m} \sum_{\ell=n+1}^m \sqrt{S_{\ell-1,k,k}}\tilde{X}_{\ell,k}\,Z_{m,\ell}
\right\},
\end{equation*}
where \(Z_{m,\ell}\) (\(\ell=1,\dotsc,m,\), \(m = n,n+1,\dotsc\)) are i.i.d. standardized Gaussian random variables.  
The following Proposition ensures positive definiteness of the random matrices 
\(\bR_{n+1},\bR_{n+2},\dotsc \) 

\begin{proposition}
If the matrix \((S_{n,k,h})_{k,h=1,\dotsc,d}\) is not singular and moreover 
\(\ba = (a_1,\dotsc,a_d)\in \BR^d\setminus \{(0,\dotsc,0)\} \) then 
\(\ba^T \bR_m \ba >0\), almost surely, for every \(m=n+1,n+2,\dotsc\)..
\end{proposition}	

\begin{proof}
Let
\begin{equation*}
W_{\ell,h}= \begin{cases}
\tilde{x}_{\ell,h} &\text{if } \ell=1,\dotsc,n,\\
\sqrt{S_{\ell-1,h,h}}\,\tilde{X}_{\ell,h} &\text{if } \ell = n+1,n+2,\dotsc	
\end{cases}
\end{equation*}	
for \(\ell=1,2,\dotsc\) and \(h=1,\dotsc,d\). In this way,
%\begin{equation*}
$S_{m,k,h}=m^{-1}\sum_{\ell=1}^{m} W_{\ell,k}\,W_{\ell,h}$. 
%\end{equation*}
Therefore for \(m=n+1, n+2,\dotsc \)
\begin{equation}\label{eq: posdef}
\begin{split}
\ba^T \bR_m \ba &= \sum_{h=1}^d \sum_{k=1}^d a_h\, a_k\, R_{m,k,h}
= \sum_{h=1}^d \sum_{k=1}^d a_h\, a_k\,
 \frac{S_{m,k,h}}{\sqrt{S_{m,k,k}\,S_{m,h,h}}}\\
&= \frac{1}{m} \sum_{l=1}^m \frac{\left(\sum_{h=1}^d a_h W_{\ell,h}\right)^2}{S_{m,h,h}}
\geq \frac{1}{m} 
\frac{\left(\sum_{h=1}^d a_h W_{n+1,h}\right)^2}{S_{m,h,h}}
\end{split}
\end{equation}
This proves that \(\bR_m\) is positive semidefinite. Being
$\sum_{h=1}^d\E(\tilde{X}_{n+1,h})=0$
we have that 
$
\E\left(\sum_{h=1}^d a_h W_{n+1,h}\right)=0,
$
and therefore the expectation of the last term in \eqref{eq: posdef} is:
\begin{equation*}
\begin{split}
\E&\left(\frac{1}{m} 
\frac{\left(\sum_{h=1}^d a_h W_{n+1,h}\right)^2}{S_{m,h,h}}\right)
= \frac{1}{mS_{m,h,h}}\var\left(\sum_{h=1}^d a_h W_{n+1,h}\right)\\
&=\sum_{h=1}^d a_h^2 S_{n,h,h} + 
\sum_{h\neq k}a_h\, a_k \, \sqrt{S_{n,h,h}\,S_{n,k,k}}\,R_{n,h,k}\\
&= \sum_{h=1}^d \sum_{k=1}^d a_h\, a_k \, S_{n,k,h}
\end{split}	
\end{equation*}
which is positive by assumption. 
The expectation of the nonnegative random variable given by \eqref{eq: posdef} 
 is positive and therefore such random variable must be positive almost surely and the proof is complete. 
\end{proof}

\vspace{0.1in}
\noindent
Proceeding in a similar way to the bivariate case, it is possible to prove that 
the random matrices \(\bS_m=(S_{m,k,h})_{k,h=1,\dotsc,d}\) and \(\bR_m\) converge almost surely, as \(m\) diverges to infinity, which yields asymptotic exchangeability of the random sequence \(\bX_1, \bX_2, \dotsc\), where \(\bX_m = (X_{m,1},\dotsc,X_{m,d})\), for \(m=1,2,\dotsc\).

\section{Illustrations}

In this section we present two examples, the first based on independent and identically distributed observations and the second based on a linear regression model.

\subsection{IID model}

We start with a simple example involving the beta distribution; i.e. the $X_{1:n}$ are i.i.d. from a beta distribution with parameters $a$ and $b$, and so the likelihood function is
$$\prod_{i=1}^n \frac{\Gamma(a+b)}{\Gamma(a)\Gamma(b)}x_i^a(1-x_i)^b.$$
A traditional Bayesian analysis is tricky due to the appearance of gamma functions of the parameters. This pretty much forces the use of a Metropolis algorithm using suitable proposal density functions for both the parameters. 

\begin{figure}
    \centering
    \includegraphics[width=14cm,height=6cm]{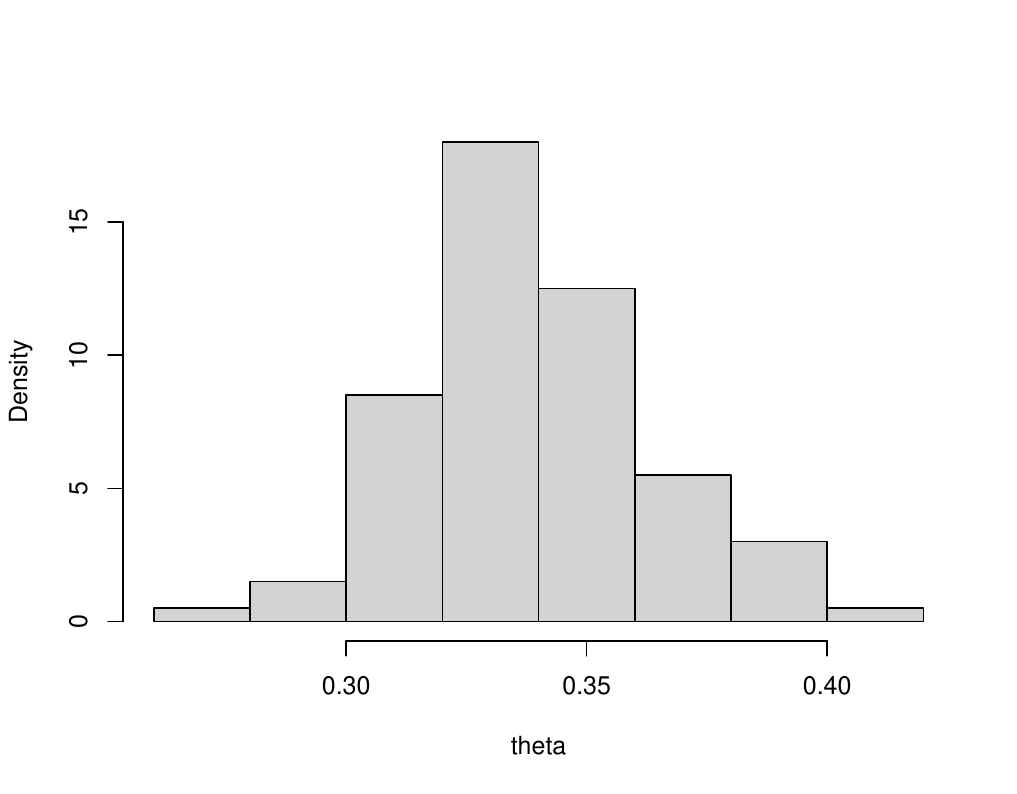}
    \caption{Mean posterior samples for the beta illustration}
    \label{fig1}
\end{figure}

\begin{figure}
    \centering
    \includegraphics[width=14cm,height=6cm]{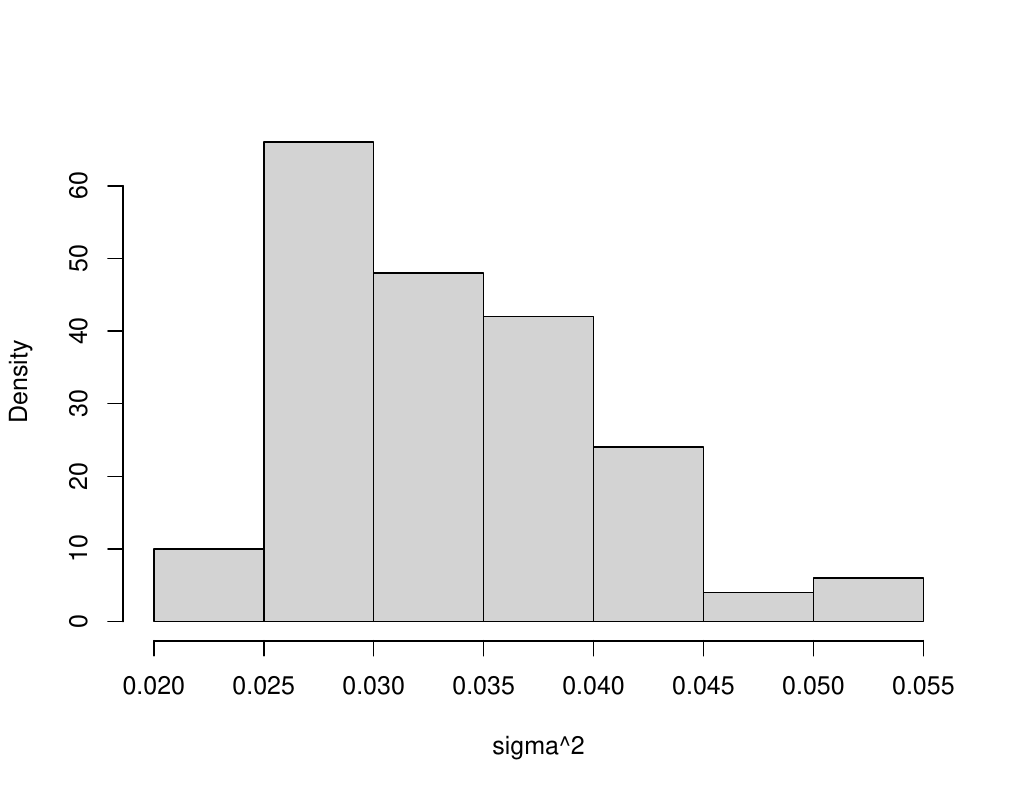}
    \caption{Variance posterior samples from the beta illustration}
    \label{fig2}
\end{figure}

\noindent
On the other hand, using the Hill predictive distributions; we can sample a missing $X_{n+1:N}$ for some suitably large $N$ in a very easy way. 
At iteration $m$ we have $X_{1:m}$ which are ordered to give $X_{(i)}$, for $i=0,\ldots,m+1$, where $X_{(0)}=1$ and $X_{(m+1)}=1$. To get $X_{m+1}$ an interval $k\in\{1,\ldots,m+1\}$ is selected with probability $1/(m+1)$ and then
$X_{m+1}$ is taken uniformly from
$(X_{(k-1)},X_{(k)})$.

We sampled $n=50$ observations from the beta distribution with parameters $(2,5)$ and took $N$ to be 1000. From each run we compute the mean, $\theta$, and variance, $\sigma^2$, of the sampled values. This would allow us to get the corresponding $a$ and $b$ values via
$$\theta=\frac{a}{a+b}\quad\mbox{and}\quad \sigma^2=\frac{ab}{(a+b)^2(a+b+1)}.$$
Based on 100 such runs, Fig.~\ref{fig1} presents the $\theta$ values and Fig.~\ref{fig2} the $\sigma^2$ values.

\begin{figure}
    \centering
    \includegraphics[width=14cm,height=6cm]{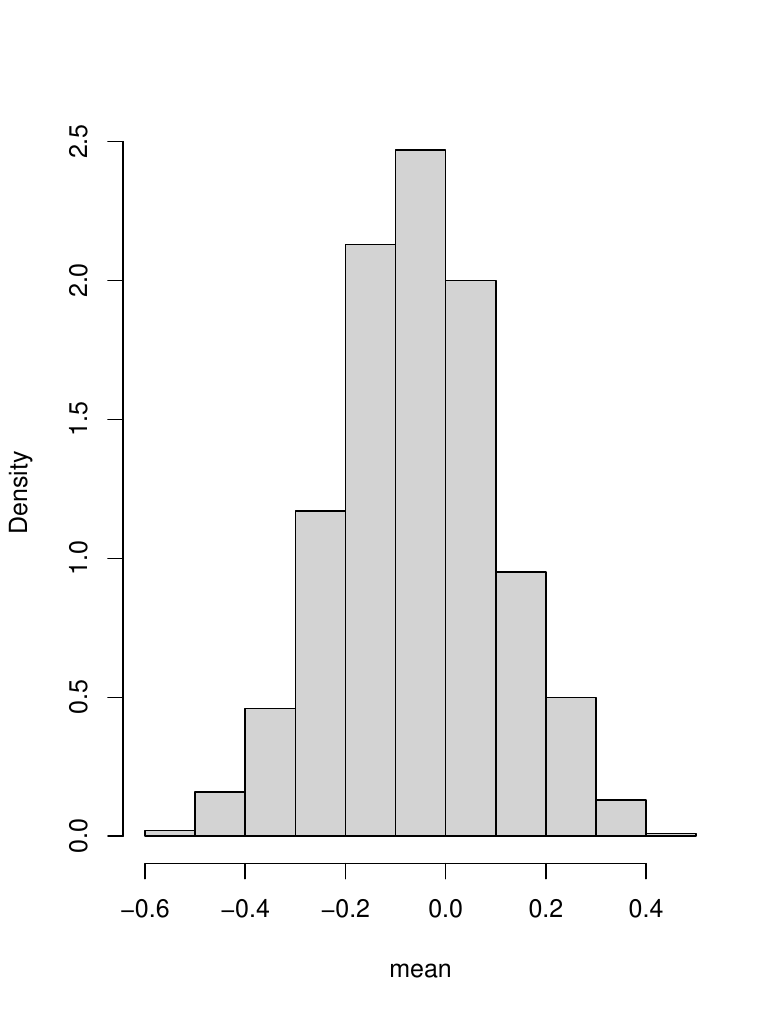}
    \caption{Mean posterior samples for the normal illustration}
    \label{fig3}
\end{figure}

\begin{figure}
    \centering
    \includegraphics[width=14cm,height=6cm]{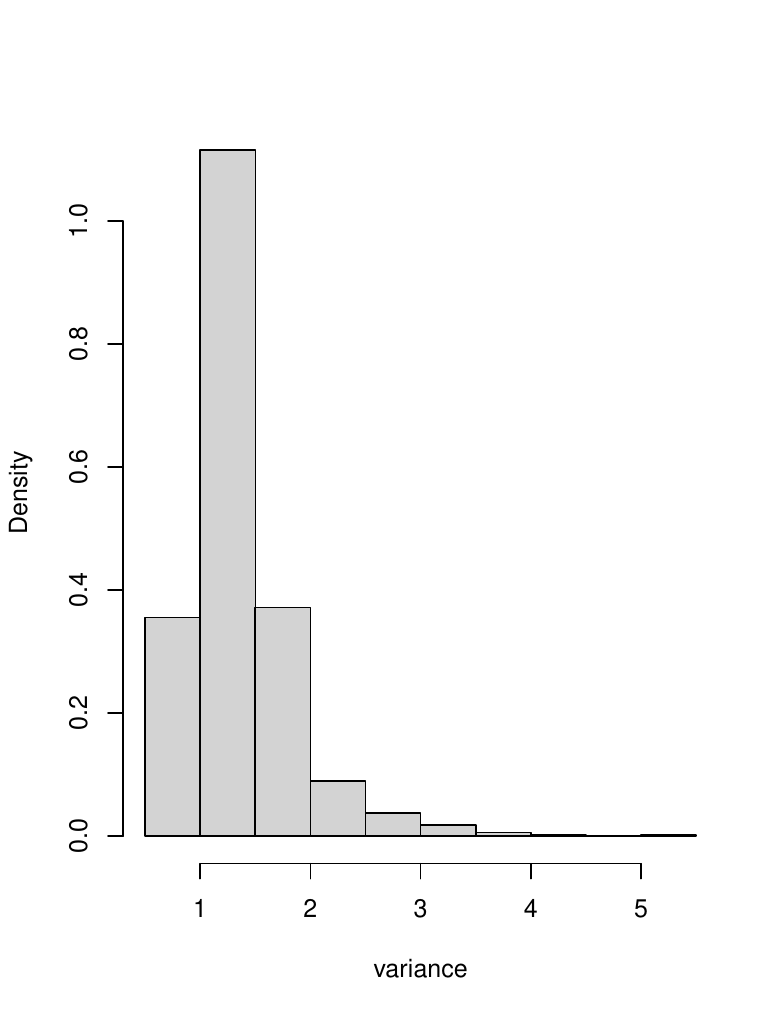}
    \caption{Variance posterior samples for the normal illustration}
    \label{fig4}
\end{figure}

One of the logistical issues to overcome is when the data are unbounded. This is best handled by using transformations, specifically to $(0,1)$. So, if the data are on the real line, we transform to $(0,1)$ using $x\to e^x/(1+e^x)$. Once we obtain $X_{1:n}$ on $(0,1)$, we can then transform back to the real line using the reverse transform $x\to \log(x/(1-x))$.
To demonstrate this, we take $n=50$ samples from the standard normal distribution function. Then 1000 sample means and variances are taken from 1000 runs each of length $N=1000$. Histograms of the posterior samples of means and variances are presented in Figs.~\ref{fig3} and \ref{fig4}, respectively.

\subsection{Linear regression model} Suppose data $(y_i,x_i)_{i=1:n}$ have been observed and a linear model of the form
$y=X\beta+\sigma\varepsilon$ is adopted where $X$ is a $n\times p$ design matrix, $\beta$ is a $p$-dimensional vector of regression parameters, $\sigma^2$ is an unknown variance parameter and $\varepsilon$ is a $n$-dimensional vector of independent standard normal random variables. For the sample of size $n$, let $\widehat\beta$ be the ordinary least squares estimator and $\widehat\sigma$ the usual estimator of the standard deviation; i.e.
$\widehat\sigma^2=||y-X\widehat\beta||^2/(n-p)$.

This gives us estimators of the error vector; i.e.
$$\widehat{e}=(y-X\widehat\beta)/\widehat\sigma.$$
These are now ordered and transformed, as in the previous illustration, to lie on $(0,1)$. The design vector $x_{n+1}$ for the to be sampled $n+1$ observation is taken uniformly from the observed design vectors; i.e. $x_{n+1}=x_i$ for $i=1,\ldots,n$, with probability $1/n$. If $\widetilde{e}_{(i)}=e^{\widehat{e}_{(i)}}/(1+e^{\widehat{e}_{(i)}})$ are the ordered transformed errors, then $\widehat{e}_{n+1}=\log(\widetilde{e}_{n+1}/(1-\widetilde{e}_{n+1}))$
where $\widetilde{e}_{n+1}$ is sampled uniformly from one of the ordered intervals with equal probability. The sampled observation $y_{n+1}$ is given by
$$y_{n+1}=x_{n+1}\widehat\beta+\widehat\sigma\,\widehat{e}_{n+1}.$$
From here the new $\widehat\beta$ and $\widehat\sigma$ can be computed based on the $n+1$ samples, the first $n$ being the observed data and the $n+1$ being the just sampled observation. Then, as has just been described, we can get $\widehat{e}_{n+2}$, and so on. For some large $N$ we can stop this process and use the $\widehat\beta$ with the $N$ samples, the first $n$ being the data and the remaining $N-n$ being sampled, as a random parameter taken from the posterior distribution. Repeating this gives multiple samples from the posterior.

\begin{figure}
    \centering
    \includegraphics[width=14cm,height=6cm]{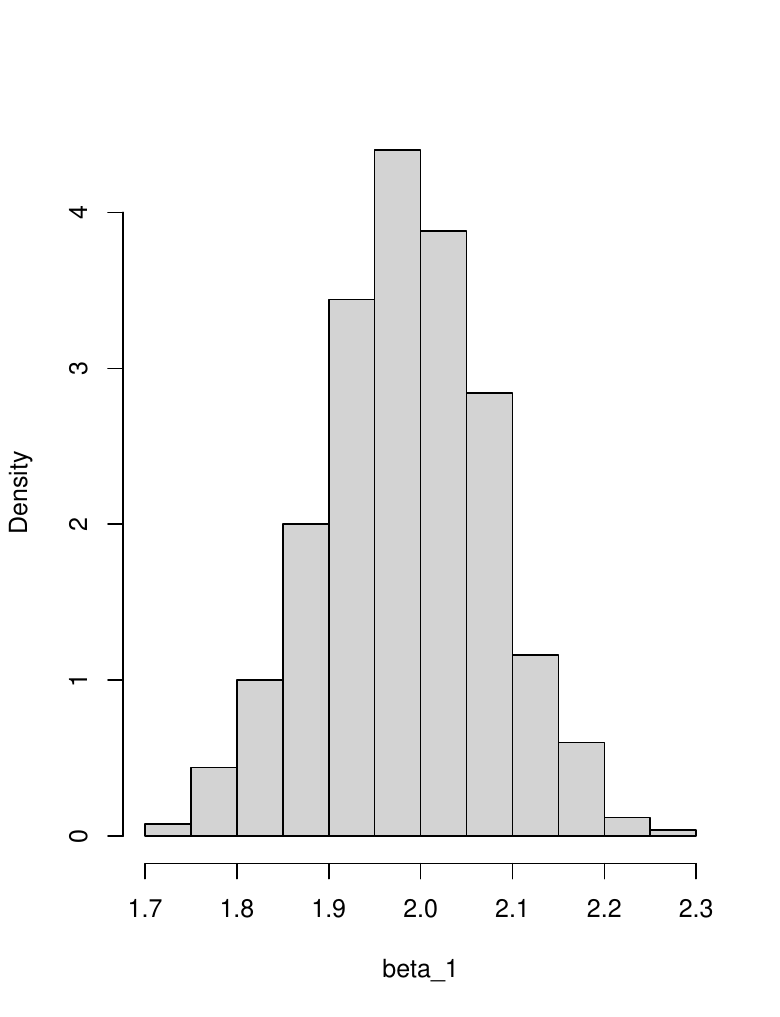}
    \caption{Posterior samples for $\beta_1$}
    \label{fig5}
\end{figure}

\begin{figure}
    \centering
    \includegraphics[width=14cm,height=6cm]{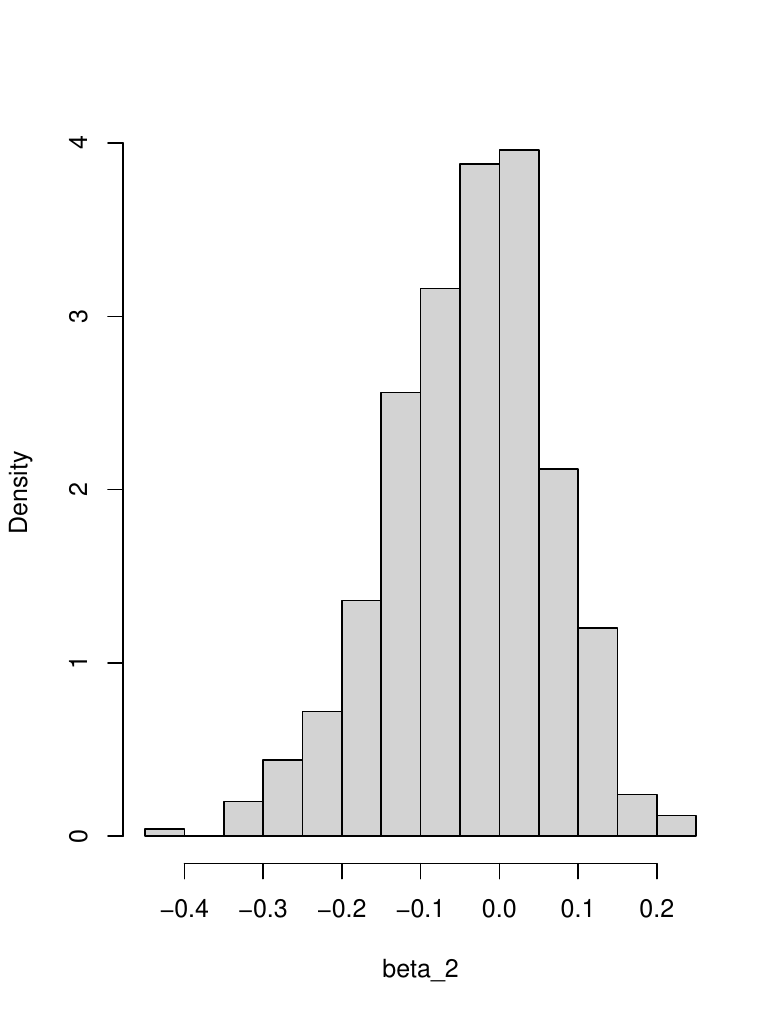}
    \caption{Posterior samples for $\beta_2$}
    \label{fig6}
\end{figure}

For an illustration, we took $n=50$ and $p=2$, generating a dataset with $\beta_1=2$ and $\beta_2=0$ and $\sigma=1$. The ordinary least squares estimator is $\widehat\beta=(1.96,-0.11)$. Posterior samples for $\beta_1$ and $\beta_2$ are given in Figs.~\ref{fig5} and \ref{fig6}, respectively. The posterior means are $(1.99,-0.04)$. 

\subsection{Bivariate model}

\begin{figure}
    \centering
    \includegraphics[width=14cm,height=6cm]{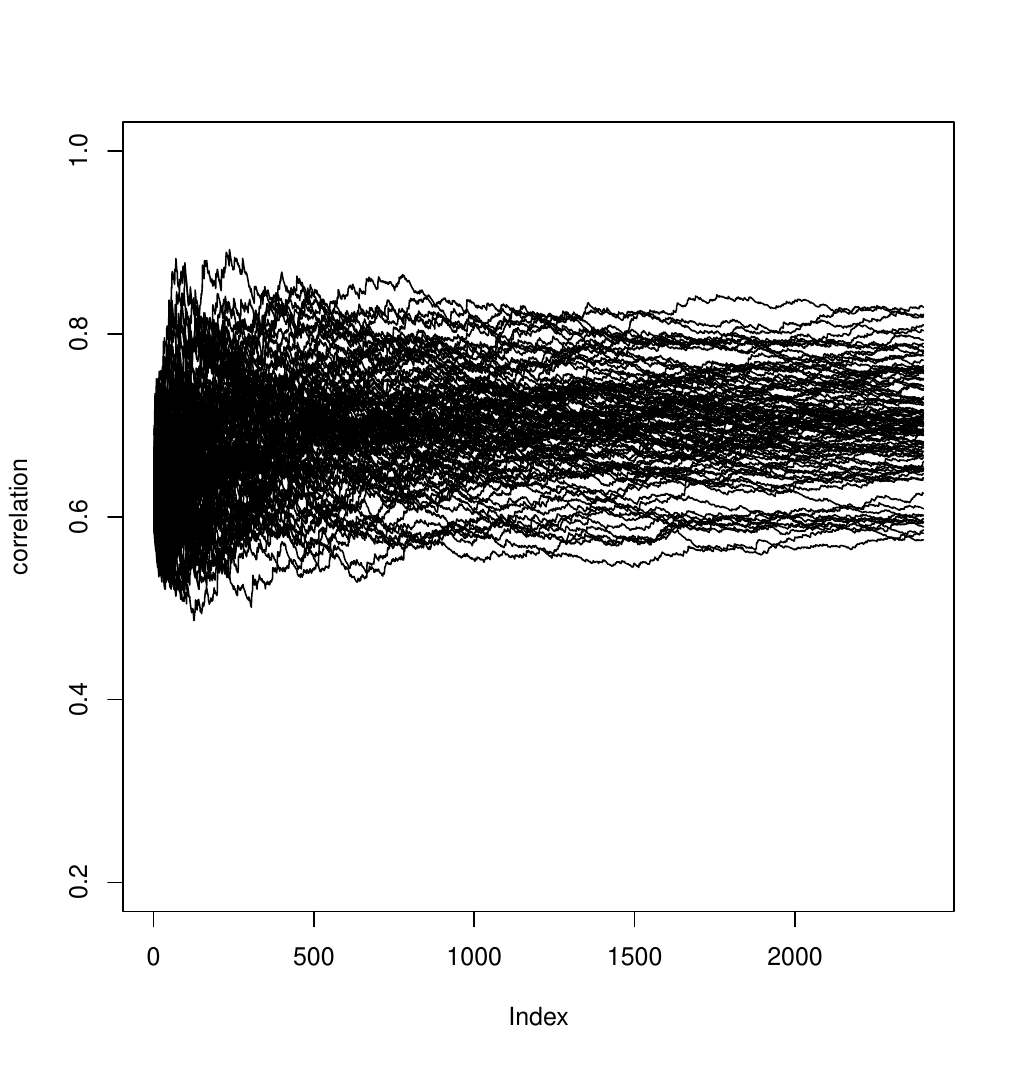}
    \caption{Posterior samples for the correlation with trajectories}
    \label{figbi}
\end{figure}

In this subsection we illustrate a bivariate model (model B). The original data are transformed to bivariate Gaussian variables $(\widetilde{x}_i,\widetilde{y}_i)$ for $i=1,\ldots,n$ with means 0 and variances 1 and with correlation $\rho$ which is estimated as
$$\widehat{\rho}_n=\frac{\sum_{i=1}^n \wt{x}_i\wt{y}_i}{\sqrt{\sum_{i=1}^n \wt{x}_i^2\sum_{i=1}^n \wt{y}_i^2}}.$$
We then generate $(\wt{x}_{n+1:N},\wt{y}_{n+1:N})$ for some large $N$ via
$$\wt{x}_{m+1}=z_{m+1}\quad\mbox{and}\quad \wt{y}_{m+1}=\widehat{\rho}_m \wt{x}_{m+1}+\sqrt{1-\widehat{\rho}_m^2}\,z'_{m+1},\quad m>n,$$
where $(z_m,z'_m)$ are independent sequences of independent standard normal random variables and 
$\widehat{\rho}_m$ is the sequence of estimators for the correlation. Fig.~\ref{figbi} shows 100 trajectories of such sequences where the original sample size is $n=100$ and the $N$ is taken to be 2500. As can be seen, each sequence $(\widehat{\rho}_m)$ converges to a limit with each limit representing a sample from the posterior.  

\section{Discussion} Essential to the predictive approach to posterior inference is the construction of a credible one step ahead predictive model. The empirical distribution has the obvious drawback of only sampling observed values - this method being known as the Bayesian bootstrap. The Hill predictive model we argue is of equal foundational importance, and is being more widely recognized under the name conformal prediction, see \cite{Papad2002}. 

The model has good properties and is easy to implement in practice. Moreover, we have introduced a novel bivariate and multivariate version of the Hill model which to our knowledge has not been achieved in such a way that in the multivariate case the marginal model is also Hill. For our model, this key property does hold.

%\bibliographystyle{agsm}
%\bibliography{hill-bib}

\begin{thebibliography}{16}
	\providecommand{\natexlab}[1]{#1}
	\providecommand{\url}[1]{\texttt{#1}}
	\expandafter\ifx\csname urlstyle\endcsname\relax
	\providecommand{\doi}[1]{doi: #1}\else
	\providecommand{\doi}{doi: \begingroup \urlstyle{rm}\Url}\fi
	
	\bibitem[Aldous(1985)]{Aldous85}
	D.~J. Aldous.
	\newblock Exchangeability and related topics.
	\newblock In \emph{\'Ecole d'\'et\'e{} de probabilit\'es de {S}aint-{F}lour,
		{XIII}---1983}, volume 1117 of \emph{Lecture Notes in Math.}, pages 1--198.
	Springer, Berlin, 1985.
	
	\bibitem[Angelopoulos and Bates(2021)]{Angel21}
	A.~Angelopoulos and S.~Bates.
	\newblock A gentle introduction to conformal prediction and distribution-free
	uncertainty quantification.
	\newblock \emph{arXiv:2107.07511}, 2021.
	
	\bibitem[Billingsley(1995)]{Billingsley95}
	P.~Billingsley.
	\newblock \emph{Probability and Measure}.
	\newblock John Wiley \& Sons, Inc., New York, third edition, 1995.
	
	\bibitem[Blackwell and MacQueen(1973)]{Blackwell1973}
	D.~Blackwell and J.~MacQueen.
	\newblock Ferguson distributions via {P}\'olya urn schemes.
	\newblock \emph{Annals of Mathematical Statistics}, 1:\penalty0 353--355, 1973.
	
	\bibitem[Doob(1949)]{Doob1949}
	J.~L. Doob.
	\newblock Application of the theory of martingales.
	\newblock \emph{Actes du Colloque International Le Calcul des Probabilites et
		ses applications, Paris CNRS}, pages 23--27, 1949.
	
	\bibitem[Ferguson(1973)]{Ferguson1973}
	T.~Ferguson.
	\newblock A {B}ayesian analysis of some nonparametric problems.
	\newblock \emph{Annals of Mathematical Statistics}, 1:\penalty0 209--230, 1973.
	
	\bibitem[Fong et~al.(2023)Fong, Holmes, and Walker]{Fong2023}
	E.~Fong, C.~Holmes, and S.~G. Walker.
	\newblock Martingale posterior distributions.
	\newblock \emph{Journal of the Royal Statistical Society, Series B},
	85:\penalty0 1357--1391, 2023.
	
	\bibitem[Gammerman et~al.(1998)Gammerman, Vovk, and Vapnik]{Gammerman98}
	A.~Gammerman, V.~Vovk, and V.~Vapnik.
	\newblock Learning by transduction.
	\newblock \emph{Uncertainty in Artificial Intelligence}, 14:\penalty0
	148–155, 1998.
	
	\bibitem[Hill(1968)]{Hill68}
	B.~M. Hill.
	\newblock Posterior distribution of percentiles: Bayes' theorem for sampling
	from a population.
	\newblock \emph{Journal of the American Statistical Association}, 63\penalty0
	(322):\penalty0 677--691, 1968.
	
	\bibitem[Papadopoulos et~al.(2002)Papadopoulos, Proedrou, Vovk, and
	Gammerman]{Papad2002}
	H.~Papadopoulos, K.~Proedrou, V.~Vovk, and A.~Gammerman.
	\newblock Inductive confidence machines for regression.
	\newblock In \emph{Machine Learning: European Conference on Machine Learning},
	pages 345--356, 2002.
	
	\bibitem[Robbins and Siegmund(1971)]{Robbins71}
	H.~Robbins and D.~Siegmund.
	\newblock A convergence theorem for non negative almost supermartingales and
	some applications.
	\newblock In J.~S. Rustagi, editor, \emph{Optimizing Methods in Statistics},
	pages 233--257. Academic Press, 1971.
	
	\bibitem[Rubin(1981)]{Rubin1981}
	D.~Rubin.
	\newblock The {B}ayesian bootstrap.
	\newblock \emph{Annals of Statistics}, 9:\penalty0 130--134, 1981.
	
	\bibitem[Sempi(2004)]{Sempi04}
	C.~Sempi.
	\newblock Convergence of copulas: critical remarks.
	\newblock \emph{Rad. Mat.}, 12\penalty0 (2):\penalty0 241--249, 2004.
	
	\bibitem[Toccaceli and Gammerman(2019)]{Tocca2019}
	P.~Toccaceli and A.~Gammerman.
	\newblock Combination of inductive mondrian conformal predictors.
	\newblock \emph{Machine Learning}, 108:\penalty0 489–510, 2019.
	
	\bibitem[Vovk(2022)]{Vovk22}
	V.~Vovk.
	\newblock In \emph{Algorithmic learning in a random world}. New York: Springer,
	2022.
	
	\bibitem[Vovk and Shafer(2008)]{Vovk08}
	V.~Vovk and G.~Shafer.
	\newblock A tutorial on conformal prediction.
	\newblock \emph{Journal of Machine Learning Research}, 9:\penalty0 371–421,
	2008.
	
\end{thebibliography}

\end{document}